\documentclass{article}

\usepackage{amsmath,amssymb,amsthm,latexsym}
\usepackage[a4paper]{geometry}
\usepackage{authblk}
\usepackage{fullpage} 
\usepackage{paralist}
\usepackage{booktabs}
\usepackage{multirow}
\usepackage{caption}
\usepackage{subcaption}
\usepackage{version}
\usepackage{enumitem} 
\usepackage{tikz}
\usetikzlibrary{positioning,arrows,automata,shapes,shadows,fit}

\newtheorem{remark}{Remark}

\newtheorem{corollary}{Corollary}
\newtheorem{lemma}{Lemma}

\newtheorem{definition}{Definition}
\newtheorem{theorem}{Theorem}

\includeversion{arxiv}
\excludeversion{lncs}

\begin{document}
 
 \title{Cryptographic Enforcement of Information Flow Policies\\ without Public Information}

 \author[*]{Jason Crampton}
 \author[*]{Naomi Farley}
 \author[*]{Gregory Gutin}
 \author[*]{Mark Jones}
 \author[**]{Bertram~Poettering}
 \affil[*]{Royal Holloway, University of London}
 \affil[**]{Ruhr University Bochum} 

 \maketitle
 
 \newcommand{\set}[1]{\left\{#1\right\}}
\newcommand{\card}[1]{\left|#1\right|}
\newcommand{\itemref}[1]{\textnormal{R\ref{#1}}}
\newcommand{\uset}[1]{\ensuremath{\mathord{\uparrow}#1}}
\newcommand{\dset}[1]{\ensuremath{\mathord{\downarrow}#1}}
\newcommand{\defeq}{\mathrel{\stackrel{\rm def}{=}}}
\newcommand{\dirpath}[3]{#1 \rightsquigarrow_{#3} #2}
\newcommand{\nodirpath}[3]{#1 \not\rightsquigarrow_{#3} #2}
\newcommand{\dom}{\set{0,1}^{\rho}}
\newcommand{\concat}{\parallel}
\newcommand{\keygen}[1]{\prf(s(#1), #1)}
\newcommand{\secgen}[2]{\prf(s(#1), #2)}

\newcommand{\eofg}{E_0^*}
\newcommand{\eofh}{E_0}

\newcommand{\setup}{\mathsf{SetUp}}
\newcommand{\derive}{\mathsf{Derive}}
\newcommand{\prf}{F}
\newcommand{\hash}{\mathsf{Hash}}

\newcommand{\getsr}{\gets_R}
\newcommand{\NN}{\mathbb{N}}
\newcommand{\keysp}{\mathcal{K}}
\newcommand{\cD}{\mathcal{D}}
\newcommand{\cA}{\mathcal{A}}
\newcommand{\Expt}{\mathrm{Expt}}
\newcommand{\Adv}{\mathrm{Adv}}
\newcommand{\kist}{{\sf kist}}
\newcommand{\ExpCorrupt}{\mathit{Corrupt}}
\newcommand{\ExpKeys}{\mathit{Keys}}


\begin{abstract}
  The enforcement of access control policies using cryptographic primitives has been studied for over 30 years.
  When symmetric cryptographic primitives are used, each protected resource is encrypted and only authorized users are given the decryption key.
  Hence, users may require many keys.
  In most schemes in the literature, keys are derived from a single key explicitly assigned to the user and publicly available information.
  Recent work has challenged this design by developing schemes that do not require public information, the trade-off being that a user may require more than one key.
  However, these new schemes, which require a chain partition of the partially ordered set on which the access control policy is based, generally require more keys than necessary.
  Moreover, no algorithm is known for determining the best chain partition to use.
  In this paper we define the notion of a tree-based cryptographic enforcement scheme, which, like chain-based schemes, requires no public information but simultaneously has lower storage requirements.
  We formally establish that the strong security properties of recent chain-based schemes are preserved by tree-based schemes, and provide an efficient construction for deriving a tree-based enforcement scheme from a given policy that minimizes the number of keys required.
\end{abstract}

 \section{Introduction}\label{sec:intro}
 
 Access control is a fundamental security service in modern computing systems.
 Informally, an access control system filters attempts by users to interact with protected resources, only allowing those interactions that are \emph{authorized} by a \emph{policy}, which is configured by the resource owner(s).
 Implementations of access control in software are vulnerable to compromise of the machine hosting the software.
 Moreover, such enforcement mechanisms do not work when protected resources are stored by an untrusted or semi-trusted third party, as is increasingly common.
  
 In some situations, therefore, we may wish to use cryptographic techniques to enforce some form of access control.
 Such an approach is useful when data objects have the following characteristics: read often, by many users; written once, or rarely, by the owner of the data; and transmitted over unprotected networks.
 In such circumstances, protected data (objects) are encrypted and authorized users are given the appropriate cryptographic keys.
 When cryptographic enforcement is used, the problem we must address is the efficient and accurate distribution of encryption keys to authorized users.

 In recent years, there has been a considerable amount of interest in \emph{key encrypting} or \emph{key assignment} schemes.
 In such schemes, a user is given a secret value -- typically a single key -- which enables the user to derive some collection of encryption keys which decrypt the objects for which she is authorized.
 Key derivation is performed using the secret value and some information made publicly available by the scheme administrator.
 These schemes are particularly suitable for policies that can be represented in terms of information flow.
 
%
Ideally, such a scheme should minimize the amount of public information and the time required to derive a key.
Unsurprisingly, it is not possible to realize both objectives simultaneously, so trade-offs have been sought.
Most schemes in the literature assume that each user is supplied with a single key from which other keys are derived with the help of some information published by the scheme administrator (see~\cite{CrMaWi06} for a survey of such schemes).
In 2010, Crampton {\em et al.}~\cite{CrDaMa10} introduced a new type of scheme in which users may receive several keys.
The significant advantage of this scheme is that no public information is required.
Moreover, the simplicity of the underlying structure of the scheme makes it possible to prove the scheme possesses very strong security properties~\cite{FrPaPo13}.
 

 An information flow policy is defined by a partially ordered set $X$ and a function mapping users and resources to elements in $X$.
 Most key assignment schemes are derived directly from $X$.
 The innovation introduced by Crampton {\em et al.} was to consider a partition of $X$ into chains (or total orders).
 It is particularly easy to work with chains, but the partition breaks some of the ``connectivity'' of the partial ordering.
 These breaks are ``repaired'' by issuing more than one key to some users.
%
 However, one question that remains open is how best to choose the chain partition of a partially ordered set: there may be many such partitions and different choices may lead to chain partition schemes with different characteristics.
 
 In this paper, we show that it is possible to work with trees, rather than chains, without reintroducing the need for public information, resulting in much more space-efficient key assignment.
 We define a tree-based, cryptographic enforcement scheme and provide a rigorous construction for such schemes from a given partially ordered set.
 We identify a number of different parameters that may be important in the context of a tree-based enforcement scheme.
 In particular, we consider the total number of keys that may be required in such a scheme and prove that a tree-based enforcement scheme with a minimal number of keys can be constructed in time $O(\card{X}^2)$.
 We show that a tree-based enforcement scheme for a given $X$ will typically require fewer keys than a chain-based scheme.
 Moreover, we present an efficient algorithm for computing the best choice of tree from the information flow policy, in contrast to chain-based methods (which assume that a chain partition is given).
 
 Our approach is based on constructing a weighted directed acyclic graph from $X$ and then constructing a minimum weight spanning out-tree from the graph.
 We establish a number of results about this out-tree that are likely to provide the foundation for further study of tree-based enforcement schemes.
 
 In the next section, we introduce notation, relevant background material and related work.
 Then, in Sec.~\ref{sec:tbes}, we define a tree-based enforcement scheme, provide a method for constructing such schemes for a given information flow policy, and prove that all the resulting schemes have the property of strong key indistinguishability.
 In Sec.~\ref{sec:minimizing-keys}, we address the problem of finding a tree-based enforcement scheme that minimizes the total number of keys required to enforce a given policy, culminating in a polynomial-time algorithm for computing such a scheme.
 We conclude the paper with a summary of our contributions and some suggestions for future work.
 Those proofs that are useful in understanding our constructions are given in the body of the paper.
 The remainder, including the security proof for our construction (which extends an earlier proof by Freire {\em et al.}~\cite{FrPaPo13}), are in the appendix.
 
 \section{Background and Related Work}\label{sec:background}

 In this paper, we consider the cryptographic enforcement of access control policies.
 In particular, we focus on the enforcement of information flow policies using symmetric cryptographic primitives.%
 \footnote{There exists a large body of work on the enforcement of attribute-based policies using asymmetric cryptographic primitives, notably attribute-based encryption~\cite{BeSaWa07,GoPaSaWa06}.}

  \subsection{Definitions and Notation}\sloppy

  A \emph{directed graph} (or \emph{digraph}) $G = (V(G),E(G))$ is defined by a \emph{vertex set} $V(G)$ and an \emph{arc set} $E(G) \subseteq V(G) \times V(G)$.
  An arc in $E(G)$ is written in the form $xy$, where $x,y \in V(G)$.
  A \emph{directed path} is a sequence of arcs $v_1 v_2,v_2 v_3,\dots,v_{p-2} v_{p-1}, v_{p-1} v_p$, which we may also write as the sequence of vertices $v_1 v_2 \dots v_p$ through which the path passes.
  We write $\dirpath{x}{y}{G}$ if there exists a directed path from $x$ to $y$ in $G$.
  For all $x \in V(G)$, we define $\dirpath{x}{x}{G}$.
  
  The \emph{in-degree} of a vertex $v \in V(G)$ is defined to be the number of arcs of the form $uv$ in $E(G)$.
  Given an undirected rooted tree, we may orient each edge in such a way that the root has in-degree $0$ and all other vertices have in-degree $1$; the resulting (acyclic) digraph is called an \emph{out-tree}.
  Thus if a directed path exists between a pair of two vertices in an out-tree then it is unique.	
  $H$ is a \emph{spanning subgraph} of a graph $G$ if $V(H) = V(G)$.
  A \emph{spanning out-tree} is a spanning subgraph that is an out-tree.
  
  A \emph{partially ordered set} or \emph{poset} is a pair $(X,\leqslant)$, where $\leqslant$ is a binary, reflexive, anti-symmetric, transitive relation.
  Given a poset $(X,\leqslant)$, we write $x < y$ if $x \leqslant y$ and $x \ne y$; and we may write $x \geqslant y$ if $y \leqslant x$.
  We write $x \lessdot y$ and say $y$ \emph{covers} $x$ if $x < y$ and there does not exist $z \in X$ such that $x < z < y$.
  We say $x$ is \emph{incomparable} to $y$, denoted $x \shortparallel y$, if $x \not\leqslant y$ and $y \not\leqslant x$.
  We say $Y \subseteq X$ is an \emph{antichain} if for all $x,y \in Y$, either $x = y$ or $x \shortparallel y$: $Y$ is a \emph{maximum} antichain if $\card{Y} \geqslant \card{Z}$ for every other antichain $Z \subseteq X$; the \emph{width} of $X$ is the cardinality of a maximum antichain.
  
  Given a poset $(X,\leqslant)$, we define the graph $H = (X,\eofh)$, where $xy \in \eofh$ if and only if $x \gtrdot y$.
  $H$ is called the \emph{Hasse diagram} of $(X,\leqslant)$ and is a directed acyclic graph.
  A Hasse diagram of a simple poset is shown in Fig.~\ref{fig:example-poset} (on page~\pageref{fig:example-poset}).
  We may also define the graph $H^* = (X,\eofg)$, where $xy \in \eofg$ if and only if $x > y$.
  The graph $H^*$ is obtained by taking the transitive closure of $H$.

  An \emph{information flow policy} is defined by a partially ordered set of security labels $(X,\leqslant)$, a set of users $U$, a set of (protected) objects $O$, and a security function $\lambda : U \cup O \rightarrow X$.
  We say $u \in U$ is \emph{authorized} to read $o \in O$ if $\lambda(u) \geqslant \lambda(o)$~\cite{BeLa76}.

 \subsection{Basic Methods of Cryptographic Enforcement}\label{subsec:basic-methods}

  A natural way to enforce an information flow policy is to define a cryptographic key $\kappa(x)$ for each $x \in X$, encrypt object $o$ with $\kappa(\lambda(o))$ and give $u$ (or enable $u$ to derive) all keys $\kappa(x)$ such that $x \leqslant \lambda(u)$.
  More specifically, let $G = (X,E(G))$ be an acyclic directed graph such that  $\eofh \subseteq E(G) \subseteq \eofg$.
  Then the transitive closure of $G$ is equal to $H^*$ and $\dirpath{x}{y}{H}$ if and only if $\dirpath{x}{y}{G}$. 
  By publishing key derivation information for each arc in $E(G)$, it is possible to derive $\kappa(y)$ from $\kappa(x)$ if $\dirpath{x}{y}{G}$.
  Thus, the total amount of key derivation information required is proportional to $|E(G)|$, while the number of key derivations will depend on the lengths of the directed paths in $G$.
  We provide a more formal account of the functionality required of a cryptographic enforcement scheme in Sec.~\ref{subsec:formalization}.
 
 Typically, key derivation information is generated using an appropriate symmetric cryptographic algorithm~\cite{AtBlFaFr09}: for arc $xy \in E(G)$, the inputs to the cryptographic algorithm will include $\kappa(x)$ and $\kappa(y)$.
 We write $Enc(m,\kappa)$ to denote the encryption of message $m$ with key $\kappa$.
 There are three very well known ways to implement cryptographic enforcement of information flow policies~\cite{CrMaWi06}:
 
 \begin{description}
  \item[Basic] -- give $u$ the set of keys $\set{\kappa(x) : x \leqslant \lambda(u)}$;
  \item[Iterative] -- give $u$ a single key $\kappa(\lambda(u))$ and publish $\set{Enc(\kappa(x),\kappa(y)) : x \lessdot y}$;
  \item[Direct] -- give $u$ a single key $\kappa(\lambda(u))$ and publish $\set{Enc(\kappa(x),\kappa(y)) : x < y}$.
 \end{description}

 We may evaluate different implementations by considering a number of parameters.
 Let $k(x)$ be the number of keys required by a user associated with $x$.
 Then we write $k$ to denote the maximum value of $k(x)$ taken over all $x$ and $K$ to denote $\sum_{x \in X} k(x)$.
 We write $p$ to denote the number of items of public information,\footnote{It is assumed that the structure of the poset $(X,\leqslant)$ is known to all participants of a cryptographic enforcement scheme.} and $d$ to denote the number of key derivation operations a user may be required to perform to derive a key.
 Let $n$ denote the cardinality of $X$.
 Then the characteristics of the three schemes described above are summarized in Table~\ref{tbl:basic-schemes}.
 
  \begin{table}[h]\setlength{\tabcolsep}{15pt}
  \[
    \begin{array}{rrrrrr}  
    \toprule
     \textbf{Scheme} & \textbf{Keys for $u$} & K & k & p & d \\
    \midrule
     \text{Basic} & ~\set{\kappa(x) : x \leqslant \lambda(u)} & ~n + \card{\eofg} & ~O(n) & 0 & 0 \\
     \text{Iterative} & \set{\kappa(\lambda(u))} & n & 1 & ~\card{\eofh} & ~O(n) \\
     \text{Direct} & \set{\kappa(\lambda(u))} & n & 1 & ~\card{\eofg} & 1 \\
    \bottomrule
    \end{array}
  \]
  \caption{How the parameters of various key assignment schemes vary}\label{tbl:basic-schemes}
  \end{table}  
 
 Naturally, there is a trade-off between the amount of public information we need to compute and store centrally, and the number of key derivation operations that are required.
 The direct scheme, for example, minimizes the cost of key derivation at the expense of an increase in public information.
 Consider the example in Fig.~\ref{fig:example-poset}: the Hasse diagram of the poset has $8$ vertices and $10$ arcs, and the width of the poset is $2$; the graph of the transitive closure has $23$ arcs. 
 
 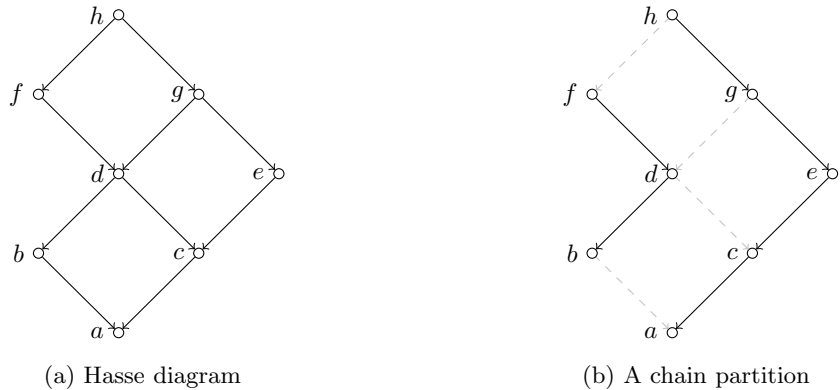
\begin{figure}[t]\centering
  \begin{subfigure}[b]{.45\textwidth}\centering
    \begin{tikzpicture}[v/.style={circle,draw,fill=white,inner sep=0pt,minimum width=4pt},<-,scale=.95,transform shape]
      \node[v,label=left:$a$] (a) {};
      \node[v,above left=of a,label=left:$b$] (b) {};
      \node[v,above right=of a,label=left:$c$] (c) {};
      \node[v,above right=of b,label=left:$d$] (d) {};
      \node[v,above right=of c,label=left:$e$] (e) {};
      \node[v,above left=of d,label=left:$f$] (f) {};
      \node[v,above right=of d,label=left:$g$] (g) {};
      \node[v,above left=of g,label=left:$h$] (h) {};
      \draw (a) -- (b);
      \draw (a) -- (c);
      \draw (b) -- (d);
      \draw (c) -- (d);
      \draw (c) -- (e);
      \draw (d) -- (f);
      \draw (d) -- (g);
      \draw (e) -- (g);
      \draw (f) -- (h);
      \draw (g) -- (h);
    \end{tikzpicture} 
   \caption{Hasse diagram}\label{subfig:hasse-diagram-illustrative}
  \end{subfigure}
  \begin{subfigure}[b]{.45\textwidth}\centering
    \begin{tikzpicture}[v/.style={circle,draw,fill=white,inner sep=0pt,minimum width=4pt},<-,de/.style={dashed,color=gray!50,very thin,<-},scale=.95,transform shape]
      \node[v,label=left:$a$] (a) {};
      \node[v,above left=of a,label=left:$b$] (b) {};
      \node[v,above right=of a,label=left:$c$] (c) {};
      \node[v,above right=of b,label=left:$d$] (d) {};
      \node[v,above right=of c,label=left:$e$] (e) {};
      \node[v,above left=of d,label=left:$f$] (f) {};
      \node[v,above right=of d,label=left:$g$] (g) {};
      \node[v,above left=of g,label=left:$h$] (h) {};
      \draw[de] (a) -- (b);
      \draw (a) -- (c);
      \draw (b) -- (d);
      \draw[de] (c) -- (d);
      \draw (c) -- (e);
      \draw (d) -- (f);
      \draw[de] (d) -- (g);
      \draw (e) -- (g);
      \draw[de] (f) -- (h);
      \draw (g) -- (h);
    \end{tikzpicture} 
   \caption{A chain partition}\label{subfig:chain-partition}
  \end{subfigure}
  \caption{The Hasse diagram of a simple poset $(X,\leqslant)$ and a chain partition}\label{fig:example-poset}
 \end{figure}
 
 More complex schemes have been devised to reduce the number of derivation operations by increasing $\card{E(G)}$~\cite{AtBlFr07,Cr11,SaFeMa08}.
 In particular, Atallah {\em et al.} introduced a scheme for policies where $X$ is a total order, in which the number of derivation operations was no greater than 2 and $\card{E(G)} = O(\card{X} \log \card{X})$~\cite{AtBlFr07}.
 Crampton extended these ideas to arbitary interval-based access control policies~\cite{Cr11}.

 \subsection{Chain Partition Techniques}\label{subsec:chain-partitions}
 
 We may consider other ways of enforcing an information flow policy.
 Crampton {\em et al.} observed that one possibility is to decompose a partially ordered set $(X,\leqslant)$ into disjoint chains and then use one-way functions to derive keys~\cite{CrDaMa10}.
 In this case, the arc set $E(G) \subseteq E_0$ and the transitive closure of $G$ (the graph representing the chain partition) is not necessarily equal to $H^*$  (as illustrated in Fig.~\ref{subfig:chain-partition}, in which deleted arcs are shown as gray dashed lines).
 
 The advantage of such a scheme is that no public information is required.
 We simply select a key for the top element in each chain and then use a (public) one-way function $F$ to iteratively compute the keys for the remaining elements in each chain.
 In particular, if $x \lessdot y$ in a chain, then $\kappa(x) = F(\kappa(y))$.%
 \footnote{This method is not appropriate for arbitrary posets because we may have $y \lessdot x$ and $y \lessdot z$~\cite{CrMaWi06}.}
 Thus a user can simply derive keys by repeated applications of the one-way function. 
 The trade-off in this case is that the user may need as many as $w$ keys, one for each of $w$ chains.
 In Fig.~\ref{subfig:chain-partition}, for example, a user assigned to vertex $d$ will require $\kappa(d)$ and $\kappa(c)$.
 In short, it may be advantageous to eliminate public information, in which case each user may require multiple keys to support key derivation.
 
 \subsection{Formalization and Constructions}\label{subsec:formalization}
 
 Recent work has formalized the security properties required of a cryptographic enforcement scheme (CES) for information flow policies~\cite{AtBlFaFr09,AtSaFeMa06,FrPaPo13}.
 Atallah {\em et al.} introduced the concepts of \emph{key recovery} and \emph{key indistinguishability}~\cite{AtBlFaFr09}.
 The former, informally, is the requirement that a coalition of users $V \subseteq U$ (the ``adversary'') can derive $\kappa(x)$ only if there exists $v \in V$ such that $\lambda(v) \geqslant x$.
 In other words, compromising users cannot lead to non-derivable keys being compromised.
 This is, essentially, the weakest security requirement that one might require of a CES.
 The schemes described in Sec.~\ref{subsec:basic-methods} have this property (provided the encryption scheme has reasonable properties).
 
 However, in the interests of integrating a CES with other cryptographic tools, the stronger notion of indistinguishability was introduced.
 This property requires that the adversary cannot distinguish between $\kappa(x)$ and a random string (of the same length).
 The schemes discussed in Sec.~\ref{subsec:basic-methods} do not have this property (see~\cite{AtBlFaFr09}, for example). 
 
 Informally, treating encryption keys as ``just another encrypted data object'' cannot be the basis for a robust cryptographic enforcement scheme.
 Specifically, the derivation of keys has to be separated from the decryption of data objects.
 We achieve this by introducing a secret value $\sigma(x)$ for each $x \in X$ from which $\kappa(x)$ may be derived.
 More formally, a CES for $(X,\leqslant)$ comprises the $\setup$ and $\derive$ algorithms, the first being used to generate keys and the data used to derive keys, and the second to derive keys.
 Let $\keysp$ denote an arbitrary key space (typically $\keysp=\set{0,1}^l$ for some $l\in\NN$). 
 \begin{itemize}
  \item $\setup$ takes as input a security parameter~$\rho$ and a poset $(X,\leqslant)$ associated with an information flow policy.
	It outputs, for each element $x \in X$, a pair $(\sigma(x), \kappa(x))$: $\sigma(x)$ is used to derive keys $\kappa(y)\in\keysp$, where $y \leqslant x$; and $\kappa(x)$ is used to encrypt data objects associated with security label~$x$.
  	The $\setup$ algorithm also outputs a set of public information $\sf Pub$, which is used to support key derivation.%
	\footnote{In some schemes, it may be the case that $\kappa(y) = \sigma(y)$ for all $y \in X$; and in some schemes, it may be that the set of public information is empty.}
  \item $\derive$ takes as input $(X,\leqslant)$, $\sf Pub$, start and end points $x,y \in X$ and $\sigma(x)$. 
	
	It outputs $\kappa(y)\in\keysp$ if and only if $y \leqslant x$.
	(In particular, $\kappa(x)$ can be derived from $\sigma(x)$.)
 \end{itemize}

 Atallah {\em et al.} described a CES in which two keys $\tau(x)$ and $\kappa(x)$ are derived from $\sigma(x)$ using a pseudorandom function and $(\tau(y),\kappa(y))$ is directly derivable from $\tau(x)$ only if $y \lessdot x$.
 (Thus, $\kappa(y)$ is iteratively derivable from $\sigma(x)$ if $\dirpath{x}{y}{}$.)
 The main innovation here is to separate the derivation and encryption functions of $\kappa(x)$, meaning that knowledge of the object decryption key $\kappa(x)$ does not help in deriving $\kappa(y)$.
 (Of course, exposure of $\tau(x)$ will allow for the derivation of $\tau(y)$ and hence $\kappa(y)$.)
 
 Freire {\em et al.} introduce a security property called \emph{strong key indistinguishability}~\cite{FrPaPo13}, which we define formally in Fig.~\ref{fig:kist} and Definition~\ref{def:kist} (on page~\pageref{def:kist}).
 The adversary selects a vertex $x$ to attack and may then learn $\set{\sigma(y) : y \not\geqslant x}$ (as in the security model for key indistinguishability) and $\set{\kappa(y) : y \ne x}$; the adversary's task is to distinguish $\kappa(x)$ from random.
 They then define a CES for total orders that has the property of strong key indistinguishability, in which a key $\kappa(x)$ is derived from $\sigma(x)$ using a pseudorandom function and $\sigma(y)$ is directly derivable from $\sigma(x)$ only if $y \lessdot x$.
 Finally, they demonstrate how this CES can be extended to arbitrary posets using the chain partition construction described in Sec.~\ref{subsec:chain-partitions}.

 \section{Tree-Based Enforcement Schemes}\label{sec:tbes}
 
 In this work, we are interested in enforcing an information flow policy, defined in terms of the Hasse diagram of a partially ordered set $(X,\leqslant)$, using cryptographic primitives.
 We may enforce the policy in any way we see fit.  
 We may, for example, increase the number of arcs (by including some subset of the transitive arcs), thereby decreasing the lengths of the directed paths in the graph and the number of key derivations that are required.  
 Thus there is a trade-off between (increasing) the number of arcs and (decreasing) the amount of storage required for public information.  
 In particular, we could include all transitive arcs, so that all paths are of length $1$ (as in the direct scheme).  
 Alternatively, we may increase the number of keys given to each user and reduce the derivation time (keeping the number of arcs constant).  
 This corresponds to allowing the user to start from multiple points in the graph.  
 
 In practice, there may be constraints that will dictate what kind of cryptographic enforcement schemes will be appropriate.
 There may be constraints, for example, on the computational power and/or storage of the end-user devices; or it may not be possible to provide an on-line server to store public information.
 As noted in Table~\ref{tbl:basic-schemes}, there are four parameters that are likely to be of interest: $k$, $K$, $p$, and $d$.
 We may wish to minimize or impose an upper bound on one or more of these parameters.
 Certain choices have been well studied, particularly those for which $k = 1$ (when each user is given exactly one key and $\eofh \subseteq E(G) \subseteq \eofg$).
%
%
 Alternatively, we can eliminate public information (by ensuring that every node has at most one in-arc), at the expense of an increase in the number of keys assigned to each vertex.
 It is these types of schemes that we consider in the remainder of this paper.
 In particular, we consider the problem of minimizing $K$,  the total number of keys required.
 
 In the special case that the Hasse diagram $H = (X,\eofh)$ is a spanning out-tree, we may use simpler cryptographic primitives to enforce an information flow policy.
 Specifically, we know there is a unique directed path from $x$ to $y$ whenever $y < x$.
 Hence, for all $x,y \in X$ such that $y \lessdot x$, we define $\kappa(y)$ to be $F(\kappa(x) \concat y)$, where $F$ is an appropriate one-way function~\cite{Sa88} and $\concat$ denotes string concatenation.
 In other words, keys are determined by the vertices, rather than the arcs, through which a directed path passes.
 In this case, we require no public information (apart from a description of the poset), because keys are derived only from a (secret) key and a (public) vertex label.

 In general, of course, $H$ is not an out-tree.
 We may assume without loss of generality, however, that our poset has a maximum element.  
 If $(X,\leqslant)$ has more than one maximal element then we add a new element to $X$ which is defined to be greater than all elements in $X$.
 (In this case, no user or object would be assigned to such an element.)
 Thus, we may assume that $H^*$ has only one vertex of in-degree zero and so has a spanning out-tree~\cite[Prop. 1.7.1]{BaGu02}.

 \subsection{Constructing an Enforcement Scheme}

 In this paper, then, we investigate ways of constructing a spanning out-tree from $H^* = (G,\eofg)$ (in order to eliminate the need for public information) by selecting an arc set that is a subset of $\eofg$. 
 However, we have to ``repair'' the Hasse diagram by allocating some users more than one key (because some of the paths will have been ``broken'' by the deletion of arcs).  
 Thus it is interesting to consider how to select the arcs for deletion in such a way that the increase in the number of keys is minimized (either on a per-vertex basis or in total).
 
 Figure~\ref{fig:example-tree-edge-deletion} illustrates three out-trees derived from the poset in Fig.~\ref{subfig:hasse-diagram-illustrative}.
 Removing arcs to create an out-tree inevitably means that certain paths are broken.
 The out-tree in Fig.~\ref{subfig:subtree-good-for-h}, for example, means that a user associated with vertex $h$ only requires a single key and derivation requires no more than one hop.
 However, every other vertex (except $a$) requires additional keys in order to bridge the gaps.
 The above observations motivate the following definition.
 
  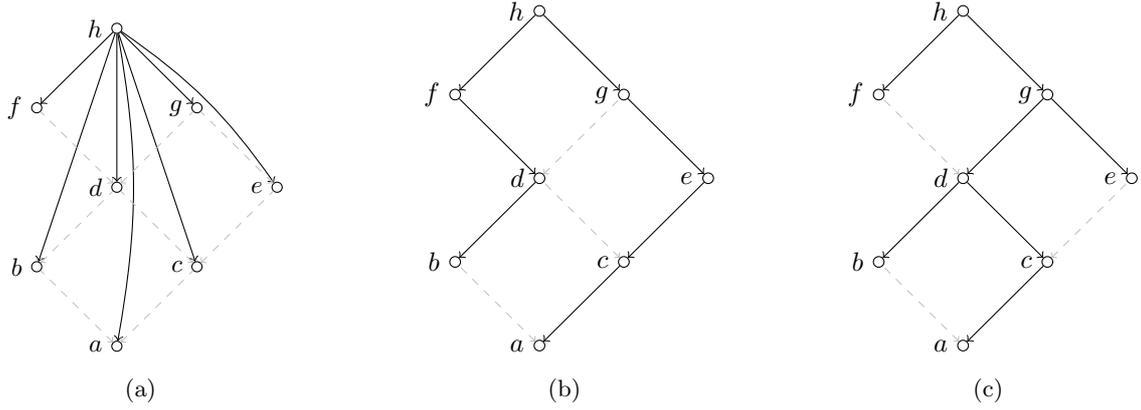
\begin{figure}[t]\centering
  \begin{subfigure}[b]{.3\textwidth}\centering 
    \begin{tikzpicture}[v/.style={circle,draw,fill=white,inner sep=0pt,minimum width=4pt},<-,de/.style={dashed,very thin,color=gray!50,<-},scale=.95,transform shape]
      \node[v,label=left:$a$] (a) {};
      \node[v,above left=of a,label=left:$b$] (b) {};
      \node[v,above right=of a,label=left:$c$] (c) {};
      \node[v,above right=of b,label=left:$d$] (d) {};
      \node[v,above right=of c,label=left:$e$] (e) {};
      \node[v,above left=of d,label=left:$f$] (f) {};
      \node[v,above right=of d,label=left:$g$] (g) {};
      \node[v,above left=of g,label=left:$h$] (h) {};
      \draw (a) to[bend right=10] (h);
      \draw (b) -- (h);
      \draw (c) -- (h);
      \draw (d) -- (h);
      \draw (e) to[bend right=10] (h);
      \draw (f) -- (h);
      \draw (g) -- (h);
      \draw[de] (a) -- (b);
      \draw[de] (a) -- (c);
      \draw[de] (b) -- (d);
      \draw[de] (c) -- (d);
      \draw[de] (c) -- (e);
      \draw[de] (d) -- (f);
      \draw[de] (d) -- (g);
      \draw[de] (e) -- (g);
    \end{tikzpicture}
    \caption{}\label{subfig:subtree-good-for-h}
  \end{subfigure}
  \hfill
  \begin{subfigure}[b]{.3\textwidth}\centering 
    \begin{tikzpicture}[v/.style={circle,draw,fill=white,inner sep=0pt,minimum width=4pt},<-,de/.style={dashed,very thin,color=gray!50,<-}]
      \node[v,label=left:$a$] (a) {};
      \node[v,above left=of a,label=left:$b$] (b) {};
      \node[v,above right=of a,label=left:$c$] (c) {};
      \node[v,above right=of b,label=left:$d$] (d) {};
      \node[v,above right=of c,label=left:$e$] (e) {};
      \node[v,above left=of d,label=left:$f$] (f) {};
      \node[v,above right=of d,label=left:$g$] (g) {};
      \node[v,above left=of g,label=left:$h$] (h) {};
      \draw[de] (a) -- (b);
      \draw (a) -- (c);
      \draw (b) -- (d);
      \draw[de] (c) -- (d);
      \draw (c) -- (e);
      \draw (d) -- (f);
      \draw[de] (d) -- (g);
      \draw (e) -- (g);
      \draw (f) -- (h);
      \draw (g) -- (h);
    \end{tikzpicture}
    \caption{}
  \end{subfigure}
  \hfill
  \begin{subfigure}[b]{.3\textwidth}\centering
    \begin{tikzpicture}[v/.style={circle,draw,fill=white,inner sep=0pt,minimum width=4pt},<-,de/.style={dashed,very thin,color=gray!50,<-}]
      \node[v,label=left:$a$] (a) {};
      \node[v,above left=of a,label=left:$b$] (b) {};
      \node[v,above right=of a,label=left:$c$] (c) {};
      \node[v,above right=of b,label=left:$d$] (d) {};
      \node[v,above right=of c,label=left:$e$] (e) {};
      \node[v,above left=of d,label=left:$f$] (f) {};
      \node[v,above right=of d,label=left:$g$] (g) {};
      \node[v,above left=of g,label=left:$h$] (h) {};
      \draw[de] (a) -- (b);
      \draw (a) -- (c);
      \draw (b) -- (d);
      \draw (c) -- (d);
      \draw[de] (c) -- (e);
      \draw[de] (d) -- (f);
      \draw (d) -- (g);
      \draw (e) -- (g);
      \draw (f) -- (h);
      \draw (g) -- (h);
    \end{tikzpicture}
    \caption{}\label{subfig:best-spanning-out-tree}
    \end{subfigure}
    \caption{Spanning out-trees derived from the poset in Fig.~\ref{fig:example-poset} by arc deletion}\label{fig:example-tree-edge-deletion}
  \end{figure}  

 \begin{definition}
  Given an information flow policy $(X,\leqslant)$, $E(T) \subseteq X \times X$ defines a \emph{derivation out-tree} $T = (X,E(T))$ if
  \begin{inparaenum}[(i)]
   \item $T$ is a spanning out-tree;
   \item $xy \in E(T)$ implies $y < x$.
  \end{inparaenum}
 \end{definition}

 \begin{lemma}\label{rem:constructing-spanning-out-tree}
  Let $D=(V,E)$ be an acyclic digraph with only one vertex $r$ of in-degree zero. 
  Then by selecting one in-bound arc for each vertex $x \ne r$ we obtain a spanning out-tree of $D$.
  Furthermore, any spanning out-tree of $D$ can be constructed in this way.
 \end{lemma}
  
 \begin{proof}
  First, let us prove that $T$ is a spanning out-tree.
  Clearly, $T$ has no directed cycle and every vertex of $x \neq r$ has in-degree $1$.
  It remains to show that $T$ is connected and contains $r$.
  Consider a vertex $y_1 \neq r$ and a longest directed path of $T$ terminating at $y_1$: $P=y_t y_{t-1} \dots y_1$.
  Since $T$ has no directed cycle all vertices of $P$ are distinct and since $P$ is longest, $y_t = r$.
  Thus, every vertex of $T$ is reachable from $r$ showing that $T$ is connected and contains $r$.

  Now let $T$ be a spanning out-tree. 
  Note that for every vertex $x \neq r$ there is exactly one arc to $x$. 
  Thus, $T$ can be constructed by the procedure of the lemma. \begin{lncs}\qed\end{lncs}
 \end{proof}
 
  If $T = (X,E)$ is a derivation out-tree and $x \ngtr u$, then $\nodirpath{x}{u}{T}$. 
  However, we may have $u < x$ but $\nodirpath{x}{u}{T}$.
  Thus, the problem with a derivation out-tree, in the context of cryptographic enforcement schemes, is that some authorized labels will no longer be reachable.
  Accordingly, we extend the notion of a derivation out-tree to a tree-based enforcement scheme.

 \begin{definition}
  Given an information flow policy $(X,\leqslant)$, a \emph{tree-based enforcement scheme} is a pair $(T,\phi)$, where $T$ is a derivation out-tree and $\phi : X \rightarrow 2^X$ is a \emph{key allocation function} such that:
  \begin{itemize}
   \item $x \in \phi(x)$;
   \item if $u \leqslant x$ then there exists $z \in \phi(x)$ such that $\dirpath{z}{u}{T}$;
   \item if $u \not\leqslant x$ then for all $z \in \phi(x)$, $\nodirpath{z}{u}{T}$.
  \end{itemize}
 \end{definition}
 
 In a tree-based enforcement scheme $(T,\phi)$, directed paths in $T$ are used to derive secrets (and hence keys): $E(T)$ determines the paths and $\phi$ determines the starting points of those paths (and hence the set of secrets that should be given to each user).
 In particular, $\phi(x) \setminus \set{x}$ is a set of vertices that were reachable from $x$ in $H^*$ that are no longer reachable in $T$.
 Thus, informally, $\phi(x)$ identifies a set of starting places in $T$ from which all (and only those) nodes that were accessible in $(X,\leqslant)$ from $x$ remain accessible in $T$, and $\card{\phi(x)} - 1$ is the number of \emph{additional} secrets that will be required by a user with security label $x$.
 
 Given a poset $(X,\leqslant)$ with maximum element $r$ and a derivation out-tree $T = (X,E)$, define $\phi_E : X \rightarrow 2^X$, where 
 \[
  \phi_E(x) = %
   \begin{cases}
    \set{x} & \text{if $x = r$}, \\
    \set{z \in X : \exists y \in X\ \text{such that}\ yz \in E, x \geqslant z, x \not\geqslant y} & \text{otherwise.}
   \end{cases}
 \]
 We now establish that $\phi_E$ is the ``best'' tree-based enforcement scheme.
 First, we show that $(T,\phi_E)$ is indeed a tree-based enforcement scheme.
 We then show that for a given tree $T = (X,E)$, any tree-based enforcement scheme $(T,\phi)$, and any $x \in X$, $\phi(x) \supseteq \phi_E(x)$.
 
 \begin{lemma}
  For any poset $(X,\leqslant)$ and any derivation out-tree $T = (X,E)$, $(T,\phi_E)$ is a tree-based enforcement scheme.
 \end{lemma}
 
 \begin{proof}
  We first show that $x \in \phi_E(x)$.
  This is trivially the case for $x = r$.
  If $x$ is not the root vertex, there exists $y \in X$ such that $yx \in E$ (since $T$ is a derivation out-tree).
  Moreover, $x \geqslant x$ and $x \not\geqslant y$ (since $yx \in E$ implies $x < y$).
  Hence, by definition, $x \in \phi_E(x)$.
  
  Now consider the case $u < x$.
  Since $T$ is a derivation out-tree, there exists a path $z_\ell z_{\ell - 1} \dots z_0$ in $T$, with $r = z_\ell$, $u = z_0$ and $\ell > 0$.
  If $z_i = x$ for some $i$ then we are done (since $x \in \phi_E(x)$).
  Hence, we may assume that $z_i \ne x$ for all $i$.
  However, there exists a smallest integer $m < \ell$ such that $x \geqslant z_m$ and $x \not\geqslant z_{m+1}$.
  (If no such integer existed, we would have to conclude $r > x$.)
  By definition, $z_m \in \phi_E(x)$ and also $\dirpath{z_m}{u}{T}$.

  Finally, consider the case $u \not\leqslant x$ and suppose (in order to obtain a contradiction) there exists $z \in \phi_E(x)$ such that $\dirpath{z}{u}{T}$.
  Then $u \leqslant z$ (by definition of a derivation out-tree and $\dirpath{}{}{T}$) and  $z \leqslant x$ (by definition of $\phi_E(x)$).
  By transitivity, $u \leqslant x$, the desired contradiction.\begin{lncs}\qed\end{lncs}
 \end{proof}

 \begin{lemma}
  For any tree-based enforcement scheme $(T = (X,E),\phi)$ and every vertex $x \in X$, $\phi(x) \supseteq \phi_E(x)$.
 \end{lemma}
 
 \begin{proof}
  Clearly $\phi(r) \supseteq \phi_E(r)$, by definition.
  Given $x \ne r$, suppose (in order to obtain a contradiction) that $z \in \phi_E(x)$ and $z \not\in \phi(x)$.
  Then, by definition of $\phi_E$, there exists $y \in X$ such that $yz \in E$, $x \geqslant z$ and $x \not\geqslant y$.
  Now, since $z \leqslant x$ and $(T,\phi)$ is an enforcement scheme, there exists $t \in \phi(x)$ such that $\dirpath{t}{z}{T}$.
  Hence $\dirpath{t}{y}{T}$ (since $T$ is a tree and $yz \in E$).
  Therefore, $y \leqslant t$ and $t \leqslant x$, since $(T,\phi)$ is an enforcement scheme and $\dirpath{t}{t}{T}$.
  By transitivity, $x \geqslant y$ (the desired contradiction).\begin{lncs}\qed\end{lncs}
 \end{proof}

 Thus, for a given tree $T$, $(T,\phi_E)$ is the enforcement scheme that minimizes, for each $x \in X$, the number of secrets required by a user assigned to $x$.
 Hence, for a given derivation out-tree $T = (X,E)$, it is reasonable to assume that we will always use the enforcement scheme $(T,\phi_E)$.
 Accordingly, we define
 \[
  K(T) = \sum_{x \in X} \card{\phi_E(x)}.
 \]
 That is $K(T)$ represents the total number of secrets required by a tree-based enforcement scheme based on the derivation out-tree $T$.
 Note also that $\card{\phi_E(x)}$ denotes the number of secrets required by a user assigned to security label $x$.
 Henceforth, given a derivation out-tree $T = (X,E)$, we will assume we will use the enforcement scheme $(T,\phi_E)$.
 Accordingly, we will write $\phi$ in preference to $\phi_E$.
 
 Let $T = (X,E)$ be a derivation out-tree.
 Then, for $y,z \in X$ such that $yz \in E$, define 
 \[ 
  \gamma(yz) = \set{x \in X : x \geqslant z, x \not\geqslant y}. 
 \]
 As we will see in Lemma~\ref{pro:computation-time-for-phi} and Sec.~\ref{sec:minimizing-keys}, there is a strong connection between $\phi$ and $\gamma$, which we can use to compute a tree-based enforcement scheme efficiently.
 
 \begin{lemma}\label{pro:computation-time-for-phi}
  Let $(X,\leqslant)$ be an information flow policy and let $T = (X,E)$ be a derivation out-tree.
  Then $\phi$ can be computed in time $O(\card{X}^2)$.
 \end{lemma}

 \begin{proof}
  By definition, $\phi(x) = \set{z \in X :  \exists y \in X\ \text{such that}\ yz \in E, x \geqslant z, x \not\geqslant y}$, for any $x$ not equal to $r$ in $X$.
  Moreover, there is a single arc in $E$ of the form $yz$, for any $z \in X$, since $T$ is a derivation out-tree.
  Thus, an algorithm to compute $\phi$ comprises an outer loop which iterates through the elements of $X$ and an inner loop that iterates through the elements of $E$, where each iteration of the inner loop for arc $yz$ tests whether $x \geqslant z$ and $x \not\geqslant y$.
  We can compute the adjacency matrix of $H^*$ in time $O(\card{X}^2)$, which we can use to test whether $x \geqslant z$ (and $x \not\geqslant y$) in constant time. 
  Moreover, $\card{E} = \card{X} - 1$ (since every vertex except the root has in-degree $1$).
  Thus our algorithm runs in time $O(\card{X}^2)$.\begin{lncs}\qed\end{lncs}
%
%
 \end{proof}

\subsection{Generating Keys}  
\label{sec:scheme}
 
 We now describe how to instantiate a tree-based enforcement scheme for $(X,\leqslant)$, given a derivation out-tree $T = (X,E)$, using a pseudorandom function (PRF).
 The scheme is a natural extension of the one used by Freire {\em et al.} for total orders~\cite{FrPaPo13}.\footnote{In the special case of a total order, we obtain the scheme of Freire {\em et al}, modulo some differences in the choice of the second input to the PRF.}
 Let $\rho$ be a security parameter and $\prf\colon \dom \times \set{0,1}^* \rightarrow \dom$ be a PRF (as formally introduced in Section~\ref{sec:proof}).
 
 \begin{description}
  \item[$\setup$:] The inputs to the algorithm are $\rho$ and a derivation out-tree $T =(X,E)$ for $(X,\leqslant)$, with root vertex $r$.  
  
		   Select secret value $s(r)$ uniformly at random from $\dom$.
		   Set  
		    \begin{align}
		     \kappa(r) &\defeq \keygen{r} \\
		    \intertext{and, recursively, if $y$ is a child of vertex $x$ (in $T$), set}
		     s(y) &\defeq \secgen{x}{y} \\
		     \kappa(y) &\defeq \keygen{y}
		    \end{align}
		    Thus, for $xy \in E$, $s(y)$ is derived from $s(x)$ and the label of~$y$, while $\kappa(y)$ is derived from $s(y)$ and the label of~$y$.
		    
		    Finally, define $\sigma(x) = \set{s(y) : y \in \phi(x)}$. 
  \item[$\derive$:] Given $y$, $x$ and $\sigma(x)$, with $y \leqslant x$, there (uniquely) exists $z \in \phi(x)$ such that $\dirpath{z}{y}{T}$.
  
		    If $z = y$, then (since $s(z) \in \sigma(x)$), compute $\kappa(z)=\keygen{z}$.
		    If $z \ne y$, then for each intermediate vertex $t_i$ on the path $t_1 \dots t_m$ between $t_1=z$ and $t_m=y$, compute $s(t_i) = \secgen{t_{i-1}}{t_i}$.
		    Finally, compute $\kappa(y)=\keygen{y}$.
 \end{description}
Our method for generating secrets is illustrated in Fig.~\ref{fig:key-gen-spanning-out-tree}.

 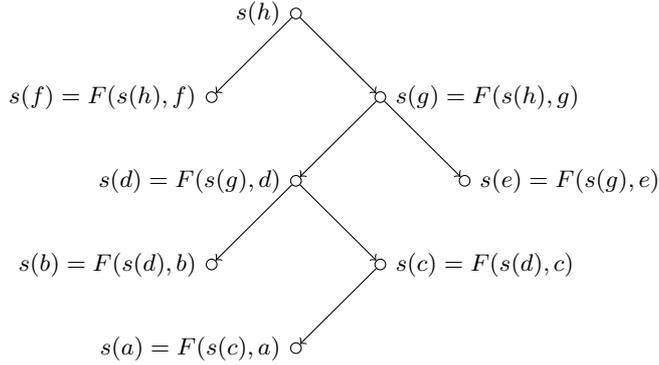
\begin{figure}[t]\centering
    \begin{tikzpicture}[v/.style={circle,draw,fill=white,inner sep=0pt,minimum width=4pt},<-]
    \node[v,label=left:{\small $s(a) = \secgen{c}{a}$}] (a) {};
    \node[v,above left=of a,label=left:{\small $s(b) = \secgen{d}{b}$}] (b) {};
    \node[v,above right=of a,label=right:{\small $s(c) = \secgen{d}{c}$}] (c) {};
     \node[v,above right=of b,label=left:{\small $s(d) = \secgen{g}{d}$}] (d) {};
    \node[v,above right=of c,label=right:{\small $s(e) = \secgen{g}{e}$}] (e) {};
    \node[v,above left=of d,label=left:{\small $s(f) = \secgen{h}{f}$}] (f) {};
    \node[v,above right=of d,label=right:{\small $s(g) = \secgen{h}{g}$}] (g) {};
    \node[v,above left=of g,label=left:{\small $s(h)$}] (h) {};
    \draw (a) -- (c);
    \draw (b) -- (d);
    \draw (c) -- (d);
    \draw (d) -- (g);
    \draw (e) -- (g);
    \draw (f) -- (h);
    \draw (g) -- (h);
  \end{tikzpicture}
  \caption{The secrets generated for the spanning-out-tree in Fig.~\ref{subfig:best-spanning-out-tree}}\label{fig:key-gen-spanning-out-tree}
 \end{figure}

\subsection{Security Analysis}\label{sec:proof} 

We start by specifying what we understand by a PRF.
Our definition is not the most general possible and is tailored to the requirements of our construction (as described in Sec.~\ref{sec:scheme}); specifically, we assume that the keyspace and range of the PRF are the same set.

\begin{definition}
  A \emph{pseudorandom function} $(\prf_\rho)_{\rho\in\NN}$ is a family of efficient functions $\prf_\rho\colon\keysp\times\set{0,1}^*\to\keysp$, where we understand $\rho$ as a security parameter and $\keysp=\set{0,1}^\rho$ as the keyspace.
\end{definition}

  We will usually write $\prf_{\rho,K}(x)$ to denote $\prf_\rho(K,x)$ for any $K\in\keysp$.
  To further simplify the notation, we will omit $\rho$ when no confusion can arise.
  We write $\cD^O\Rightarrow 1$ to denote a configuration where $\cD$ is a probabilistic poly-time Turing machine that has oracle access to a function $O$ and outputs a bit with value~$1$.

  \begin{definition}
  Given a pseudorandom function $\prf$, we define the \emph{advantage} of a distinguisher~$\cD$ to be   
  \[
  \Adv^\prf_\cD(\rho)=\left\lvert\Pr[K\getsr\keysp;\cD^{\prf_K(\cdot)}\Rightarrow 1]-\Pr[\varphi\getsr\langle\set{0,1}^*\to\keysp\rangle;\cD^{\varphi(\cdot)}\Rightarrow 1]\right\rvert\, ,
  \]
  where $\langle\set{0,1}^*\to\keysp\rangle$ denotes the universe of all functions mapping $\set{0,1}^*$ to~$\keysp$.
  We say $\prf$ is \emph{indistinguishable} from a random function if the advantage of any efficient distinguisher~$\cD$ is negligible.
\end{definition}

We next make precise the level of security that we target.
We refer to~\cite{AtBlFaFr09,FrPaPo13} for recent discussions and comparisons of security models that are specific enough to allow the analysis of CESs using the formalisms of provable security. 
We reproduce here the strongest model from~\cite{FrPaPo13}; that is, the one formalising the highest level of security, which is based on the security experiment $\Expt^{\kist,b}_{X,x,\cA}(1^\rho)$ defined in Fig.~\ref{fig:kist}.
We write $\bar{\sigma}$ and $\bar\kappa$ to denote, respectively, vectors that list the values $\sigma(x)$ and~$\kappa(x)$ for all $x\in X$.

\begin{figure}[t]
  \centering
    \begin{minipage}[t]{8.5cm}
      $\Expt^{\kist,b}_{X,x,\cA}(1^\rho)$:
    \vspace*{-1.5mm}
\setlist[enumerate,1]{
  align=left,
  leftmargin=0pt,itemindent=!,labelindent=0pt,labelwidth=2.2em,labelsep=0pt,
  label={\protect\makebox[1em][r]{\footnotesize\texttt{\protect\color{gray}\arabic{enumi}}}\hspace*{0.4em}},
  ref={\arabic{enumi}}
}
    \begin{enumerate}
    \item $({\sf Pub},\bar{\sigma},\bar{\kappa})\getsr\setup(1^\rho,(X,\leqslant))$
    \item $K_0\getsr\keysp$
      \label{lin:K0}
    \item $K_1\gets \kappa(x)$
      \label{lin:K1}
    \item $b'\getsr\cA((X,\leqslant),x,{\sf Pub},\ExpCorrupt_{X,x},\ExpKeys_{X,x},K_b)$
    \item Return $b'$
    \end{enumerate}
  \end{minipage}%
  \caption{Security experiment for strong key indistinguishability}
  \label{fig:kist}
\end{figure}

\begin{definition}
  \label{def:kist}
  Let $(X,\leqslant)$ be an arbitrary poset.
  A CES for $(X,\leqslant)$ is \emph{strongly key indistinguishable with respect to static adversaries} if, for all $x\in X$, the advantage of all efficient adversaries~$\cA$ that interact in experiment $\Expt^{\kist}_{X,x,\cA}$ is negligible, where we define
  \[
  \Adv^{\kist}_{X,x,\cA}(\rho)=\left\lvert\Pr\left[\Expt^{\kist,1}_{X,x,\cA}(1^\rho)\Rightarrow1\right]-\Pr\left[\Expt^{\kist,0}_{X,x,\cA}(1^\rho)\Rightarrow1\right]\right\rvert
  \]
  and set $\ExpCorrupt_{X,x}=\{\sigma(v): v\in X, x\not\leqslant v\}$ and $\ExpKeys_{X,x}=\{\kappa(v):v\in X\setminus\{x\}\}$.
\end{definition}

Observe that in this definition, and in contrast to other models discussed in~\cite{AtBlFaFr09,FrPaPo13}, the adversary obtains, in principle, \emph{all} secrets embedded in the system (that is, all $\sigma(x)$ and $\kappa(x)$ values), excluding only those that would allow distinguishing the target key by trivial means (e.g., by invoking the $\derive$ algorithm).%
  \footnote{A variant of Definition~\ref{def:kist} would consider dynamic adversaries: such an adversary is able to choose the challenge label $x$ \emph{during} the experiment, rather than having it fixed as one of the experiment's parameters. 
	    However, it has been shown that static and dynamic definitions of key indistinguishability are polynomially equivalent~\cite{FrPaPo13}.
	    To simplify the exposition, therefore, we restrict our attention to the static case.}

The final step of our analysis is to prove that our tree-based enforcement scheme from Sec.~\ref{sec:scheme} is strongly key indistinguishable.
Observe that this implies that our scheme is secure in all the models considered in~\cite{AtBlFaFr09,FrPaPo13}.
More formally, we have the following result.

\begin{theorem}\label{thm:security}
  Our tree-based enforcement scheme is strongly key indistinguishable in the sense of Definition~\ref{def:kist}.
  More precisely, for any poset $(X,\leqslant)$, $x\in X$, and efficient adversary~$\cA$, there exists a constant $0\leqslant c\leqslant\lvert X\rvert$ and efficient distinguishers $\cD^0_1,\ldots,\cD^0_c$, $\cD^1_1,\ldots,\cD^1_c$ against the underlying PRF such that
  \[
  \Adv^{\kist}_{X,x,\cA} \quad\leqslant\quad \Adv^\prf_{\cD^0_1}+\dots+\Adv^\prf_{\cD^0_c}+\Adv^\prf_{\cD^1_1}+\dots+\Adv^\prf_{\cD^1_c}
  \enspace.
  \]
\end{theorem}

 \section{Minimizing $K$ in a Tree-based Enforcement Scheme}\label{sec:minimizing-keys}
 
 So far, we have shown that it is possible to construct a tree-based enforcement scheme for an information flow policy $(X,\leqslant)$ that is strongly key indistinguishable.
 As we observed before, we will usually require our tree-based enforcement scheme to have some particular properties, such as minimizing the total number of keys or ensuring that all derivation paths are no longer than some threshold value.
 Hence, we require an algorithm to compute a derivation out-tree that satisfies the desired requirements, since, by Lemma~\ref{pro:computation-time-for-phi}, we can then compute the associated key allocation function $\phi$ in polynomial time.

 In this section, we consider two questions: how to minimize $K$, the total number of keys allocated to vertices (by the key allocation function $\phi$); and how to minimize $\widehat{K}$, the total number of keys distributed to users.
 The second question is interesting because, in practice, we might want to reduce the exposure of keys by ensuring that very few keys are associated with vertices to which many users are assigned.
 We solve both questions, demonstrating that it is surprisingly efficient to compute the required tree-based enforcement schemes in polynomial time.
 This is possible because of the connection between $\phi$ and $\gamma$, which leads to Theorem~\ref{thm:number-of-keys-equals-sum-of-gammas}.
 We then state and prove Theorem~\ref{thm:computing-tree-based-enf-scheme}, the main result of this section.
 
 Our basic approach is to define a weight for each arc in $\eofg$ and construct a minimum weight spanning out-tree.
 Accordingly, given an information flow policy $((X,\leqslant),\lambda,U,O)$, where \mbox{$\lambda : U \cup O \rightarrow X$}, let $U(x) = \set{u \in U : \lambda(u) = x}$, and let \mbox{$H = (X,\eofh)$} be the Hasse diagram of $X$.
 Then we define the \emph{weight function} \mbox{$\omega : \eofg \rightarrow \mathbb{N}$}, where
 \[
    \omega(yz) \defeq \sum_{x \in \gamma(yz)}\card{U(x)}.
 \]
  
  \begin{theorem}\label{thm:number-of-keys-equals-sum-of-gammas}
  Let $(T=(X,E),\phi)$ be any tree-based enforcement scheme for $(X,\leqslant)$.
  Then
  \[
   \sum_{\stackrel{x \in X}{x \ne r}} \card{U(x)} \cdot \card{\phi(x)} = \sum_{e \in E} \omega(e).
  \]
 \end{theorem}

 \begin{proof}
  By definition, we have, for every $x \ne r$,
  \begin{align*}
   \card{\phi(x)} &= \card{\set{yz \in E : x \in \gamma(yz)}}
  \intertext{and so}
   \card{U(x)} \cdot \card{\phi(x)} &= \card{U(x)} \cdot \card{\set{yz \in E : x \in \gamma(yz)}}.
  \intertext{Hence}
   \sum_{\stackrel{x \in X}{x \ne r}} \card{U(x)} \cdot \card{\phi(x)} %
				 &= \sum_{\stackrel{x \in X}{x \ne r}} \card{U(x)} \cdot \card{yz \in E : x \in \gamma(yz)}
  \intertext{and, since $r \not\in \gamma(yz)$ for any $yz \in E$, we have}
   \sum_{\stackrel{x \in X}{x \ne r}} \card{U(x)} \cdot \card{\phi(x)} %
				 &= \sum_{yz \in E}  \sum_{x \in \gamma(yz)} \card{U(x)} = \sum_{yz \in E} \omega(yz).
  \end{align*}
  \begin{lncs}\qed\end{lncs}
 \end{proof}

 \begin{theorem}\label{thm:computing-tree-based-enf-scheme}
  Given an information flow policy $((X,\leqslant),U,O,\lambda)$, we can compute a tree-based enforcement scheme $(T,\phi)$ such that $\widehat{K}$ is minimized in time \mbox{$O(\card{\eofg} + \card{X}^2)$}.
 \end{theorem}

 \begin{proof}
  For brevity, we write $E$ for $E(T)$.
  By Theorem~\ref{thm:number-of-keys-equals-sum-of-gammas}, 
  \[
   \widehat{K} = \card{U(r)} + \sum_{e \in E} \omega(e).
  \]
  An algorithm to compute the weight function $\omega$ iterates through the arcs in $\eofg$ and, for a given arc $yz$, iterates through all $x$ in $X$ testing whether $x \geqslant z$ and $x \not\geqslant y$.
  In other words, we swap the inner and outer loops in the algorithm used in the proof of Lemma~\ref{pro:computation-time-for-phi}.
  Thus, we can compute $\omega$ in time $O(\card{X}^2)$.
  
  Since $\card{U(r)}$ is fixed, we minimize $\widehat{K}$ by computing a derivation out-tree that minimizes $\sum_{e \in E} \omega(e)$.
  By Lemma~\ref{rem:constructing-spanning-out-tree}, we can achieve this by selecting, for each non-root vertex $x \in X$, the minimum weight arc to $x$, where the weights are given by $\omega$.
  We need only consider each arc (in $\eofg$) once, which takes time $O(\card{\eofg})$.
  The resulting set of arcs forms a spanning out-tree of minimum weight and the number of additional keys required is $\sum_{e \in E} \omega(e)$.
  We can derive the associated key allocation function in time $O(\card{X}^2)$, by Lemma~\ref{pro:computation-time-for-phi}; the result follows.\begin{lncs}\qed\end{lncs}
 \end{proof}
  
 \begin{corollary}\label{cor:reducing-run-time}
  Given an information flow policy $((X,\leqslant),U,O,\lambda)$, we can compute a tree-based enforcement scheme such that $K$ is minimized in time $O(\card{\eofg} + \card{X}^2)$.
 \end{corollary}
%

 \begin{corollary}\label{cor:minimizing-leaves}
  We can find, in time $O(\card{\eofg}+\card{X}^{3/2}\card{\eofg}^{1/2})$, a minimum weight spanning out-tree that has the minimum number of leaves among such trees.   
 \end{corollary}
%
 
 It is useful to find a minimum weight spanning out-tree with a minimum number of leaves because the number of leaves will impose an upper bound on $\card{\phi(x)}$.
 Note, however, that $\card{\phi(x)}$ may be greater than the width of $X$ (and it is not difficult to construct such an example).
 This is because the set of arcs in the graph that is input to {\sc MinLeaf}~--~the algorithm used to construct the spanning out-tree~--~will, in general, be a strict subset of $\eofg$.
 Thus, the size of the maximal independent set in the graph that is input to {\sc MinLeaf} can exceed the width of the poset (which is the equal to the size of the maximal independent set in $G = (X,\eofg)$).  

 We now prove some further properties of $\gamma$.
 This enables us to reduce the running time of our algorithm because we show it is sufficient to consider only arcs in $\eofh$ (rather than $\eofg$) when constructing the minimum weight spanning out-tree.
 
 \begin{lemma}\label{lem:gamma-properties}
  Let $(X,\leqslant)$ be a partially ordered set.
  Then for all $x,y,z \in X$ such that $z <y < x$, 
    \[
     \gamma(xy) \cap \gamma(yz) = \emptyset \quad\text{and}\quad \gamma(xz) \supseteq  \gamma(yz) \cup \gamma(xy)
    \]
 \end{lemma}
%
%
 
 \begin{corollary}\label{cor:weights-of-transitive-edges}
  Let $(X,\leqslant)$ be a partially ordered set with Hasse diagram \mbox{$H = (X,\eofh)$}.
  Then, for any path $x_1 x_2 \dots x_p$ in $H^*$, $p > 2$, we have 
    \[
     \omega(x_1 x_p) \geqslant \sum_{i=1}^{p-1} \omega(x_i x_{i+1}).
    \]
 \end{corollary}
%

 \begin{corollary}\label{cor:only-need-covering-relation-for-min-spanning-tree}
  Let $(X,\leqslant)$ be a partially ordered set with Hasse diagram \mbox{$H = (X,\eofh)$}.
  Then there exists a minimum weight spanning out-tree $T = (X,E)$ with $E \subseteq \eofh$.
 \end{corollary}
%
 
 \begin{corollary}\label{cor:reducing-run-time-further}
  We can compute a tree-based enforcement scheme for information flow policy $(X,\leqslant)$ in time $O(\card{\eofh} + \card{X}^2)$.
 \end{corollary}
%
 
 \begin{remark}
  In practice, we expect that $\card{U(x)} > 0$, although our proofs do not make this assumption.
  If we do make this assumption, it is possible to strengthen the statement in Corollary~\ref{cor:only-need-covering-relation-for-min-spanning-tree} and assert that a minimum weight spanning out-tree can \emph{only} contain arcs from the Hasse diagram.
 \end{remark}

 Fig.~\ref{fig:minimum-weight-spanning-tree} illustrates the construction of the minimum weight spanning out-tree for the poset in Fig.~\ref{fig:example-poset} (assuming there is a single user for each vertex).
 The weight on arc $ec$ is $3$, for example, because $\gamma(ec) = \set{c,d,f}$.
 (The effect of retaining arc $ec$ would be that $\kappa(c)$ would be required for each of $c$, $d$ and $f$.
 Equivalently, $c \in \phi(d)$ and $c \in \phi(f)$ if we were to choose $ec$ to belong to our derivation out-tree.)
 To construct a minimum weight spanning out-tree, we must select arcs $ca$ and $dc$ (and we select one or other of $fd$ and $gd$).
 One possible scheme, when $gd$ is retained rather than $fd$ is illustrated in Fig.~\ref{subfig:derivation-tree}; the scheme requires a total of $11$ keys, being the sum of the weights on the retained arcs plus an extra one for the root vertex.

 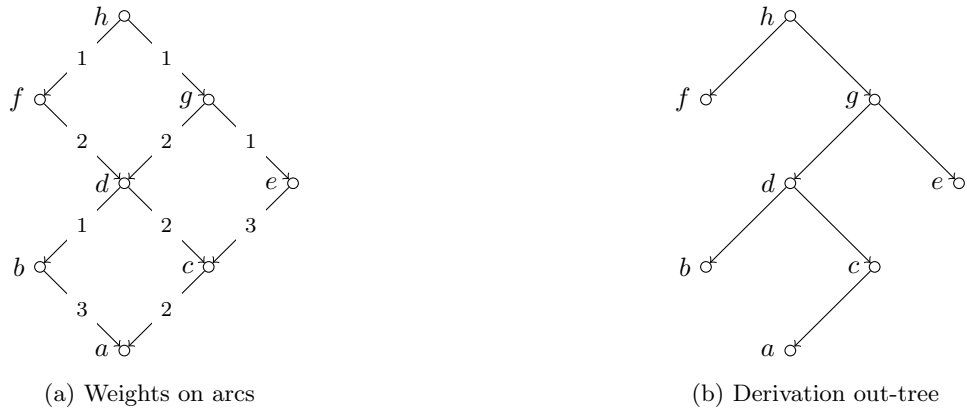
\begin{figure}[h]\centering
  \begin{subfigure}[b]{.45\textwidth}\centering
    \begin{tikzpicture}[v/.style={circle,draw,fill=white,inner sep=0pt,minimum width=4pt},<-]
      \node[v,label=left:$a$] (a) {};
      \node[v,above left=of a,label=left:$b$] (b) {};
      \node[v,above right=of a,label=left:$c$] (c) {};
      \node[v,above right=of b,label=left:$d$] (d) {};
      \node[v,above right=of c,label=left:$e$] (e) {};
      \node[v,above left=of d,label=left:$f$] (f) {};
      \node[v,above right=of d,label=left:$g$] (g) {};
      \node[v,above left=of g,label=left:$h$] (h) {};
      \draw (a) to node[fill=white] {\footnotesize $3$} (b);
      \draw (a) to node[fill=white] {\footnotesize $2$} (c);
      \draw (b) to node[fill=white] {\footnotesize $1$}  (d);
      \draw (c) to node[fill=white] {\footnotesize $2$} (d);
      \draw (c) to node[fill=white] {\footnotesize $3$} (e);
      \draw (d) to node[fill=white] {\footnotesize $2$} (f);
      \draw (d) to node[fill=white] {\footnotesize $2$} (g);
      \draw (e) to node[fill=white] {\footnotesize $1$} (g);
      \draw (f) to node[fill=white] {\footnotesize $1$} (h);
      \draw (g) to node[fill=white] {\footnotesize $1$} (h);
    \end{tikzpicture} 
   \caption{Weights on arcs}\label{subfig:hasse-diagram-weights}
  \end{subfigure}
 \hfill
 \begin{subfigure}[b]{.45\textwidth}\centering
  \begin{tikzpicture}[v/.style={circle,draw,fill=white,inner sep=0pt,minimum width=4pt},<-]
    \node[v,label=left:$a$] (a) {};
    \node[v,above left=of a,label=left:$b$] (b) {};
    \node[v,above right=of a,label=left:$c$] (c) {};
    \node[v,above right=of b,label=left:$d$] (d) {};
    \node[v,above right=of c,label=left:$e$] (e) {};
    \node[v,above left=of d,label=left:$f$] (f) {};
    \node[v,above right=of d,label=left:$g$] (g) {};
    \node[v,above left=of g,label=left:$h$] (h) {};
    \draw (a) -- (c);
    \draw (b) -- (d);
    \draw (c) -- (d);
    \draw (d) -- (g);
    \draw (e) -- (g);
    \draw (f) -- (h);
    \draw (g) -- (h);
  \end{tikzpicture}
  \caption{Derivation out-tree}\label{subfig:derivation-tree}
  \end{subfigure}
  \caption{The minimum weight derivation tree for Fig.~\ref{fig:example-poset}}\label{fig:minimum-weight-spanning-tree}
 \end{figure}
 
 
 \begin{remark}
  Our construction will almost always require fewer keys than a scheme based on chain partitions.
  This follows by noting that any vertex $x$, such that $x > y$, $x > z$ and $\set{y,z}$ is an antichain, necessarily requires (at least) two keys in a chain partition scheme, but this is not necessarily true of our construction (since the derivation tree may include many antichains).  
  Consider the chain partition in Fig.~\ref{subfig:chain-partition} and the derivation tree in Fig.~\ref{subfig:derivation-tree}.
  The former would require $13$ keys, while the latter only $11$.
 \end{remark}

 \section{Conclusion}\label{sec:conclusion}
 
 In this paper, we have introduced a new form of cryptographic scheme for the enforcement of information flow policies.
 Our scheme has the advantage that no public information is required for the derivation of decryption keys.
 Moreover, our tree-based scheme requires fewer keys (when $X$ is not a total order), compared to existing chain-based approaches, to enforce a given policy.
 Nevertheless, our scheme retains the strong security properties that have recently been established for chain-based schemes~\cite{FrPaPo13}.
 From a practical perspective, we provide an efficient algorithm for computing an optimal derivation tree, in the sense that it requires the smallest number of keys.
 This is in sharp contrast to chain-based approaches, which provide no guidance on how best to select a chain partition of the poset (of which there may be many) nor provide a way of computing the number of keys required for a given partition.
 Thus, there are particular practical advantages to using a tree-based approach.
 
 There are several interesting opportunities for future work.
 From a mathematical perspective, it would be interesting to establish the minimum total number of keys required by a chain-based scheme and, if possible, to quantify the benefits offered by a tree-based scheme.
 This is, however, likely to be non-trivial, as it is not clear that there exists a weight function for chain-based schemes that can be used to formulate a result analogous to Theorem~\ref{thm:number-of-keys-equals-sum-of-gammas}. 
 From a more practical perspective, it would be interesting to find an algorithm that can compute a derivation tree such that%
 \begin{inparaenum}[(i)]%
  \item no user requires more than $w$ keys, where $w$ is the width of the poset
  \item the total number of keys is as small as possible.
 \end{inparaenum}
 In particular, such a construction may be useful in scenarios where the user devices have limited secure storage for keys.
 Our preliminary work on this problem suggests that no efficient algorithm exists, but whether it is an NP-hard problem remains open.
 We also intend to investigate whether a forest-based enforcement scheme, which would share some of the characteristics of tree- and chain-based schemes, would offer advantages in terms of reducing%
 \begin{inparaenum}[(i)]%
  \item the maximum number of steps required for key derivation
  \item the administrative effort required following key revocation (since we can limit key updates to those vertices within a tree in the forest).
 \end{inparaenum}
 In Fig.~\ref{subfig:derivation-tree}, for example, we could delete arc $gd$ to yield a forest of two trees: each user assigned to vertex $h$ or $g$ would require an additional key ($\kappa(d)$) but worst-case key derivation would require two, rather than four, hops.

\paragraph{Acknowledgments.} 

BP was supported by EPSRC Leadership Fellowship EP/H005455/1, a Sofja Kovalevskaja Award of the Alexander von Humboldt Foundation, and the German Federal Ministry for Education and Research.

 \bibliography{refs}
 \bibliographystyle{abbrv}

\end{document}